\documentclass[runningheads]{llncs}
\usepackage[T1]{fontenc}
\usepackage{graphicx}
\usepackage[authoryear]{natbib}
\usepackage{amsmath}
\usepackage{amsfonts}
\usepackage{amssymb}
\usepackage{cleveref}
\usepackage{graphicx}
\usepackage{mathtools}
\usepackage[textsize=tiny]{todonotes}
\usepackage{url}
\usepackage{pgfplots}

\provideboolean{submissionversion}
\setboolean{submissionversion}{true}

\provideboolean{camerareadyversion}

\ifthenelse{\boolean{submissionversion}}{
  \renewcommand{\todo}[2][1]{}
  \newcommand{\newedit}[1]{#1}
}{
  \newcommand{\newedit}[1]{{\color{green} #1}}
}

\newcommand{\I}[1]{\ensuremath{\mathbb{I}_{\left\{#1\right\}}}} 
\newcommand{\E}{\ensuremath{\mathbb{E}}}

\newcommand{\paraheader}[1]{\smallskip\noindent{\sffamily\bfseries #1}}

\begin{document}

\title{Latency Advantages in Common-Value Auctions}

\provideboolean{anonymousversion}

\ifthenelse{\boolean{anonymousversion}}{
  \author{Anonymized}
  \authorrunning{Anonymous}
}{
\author{Ciamac Moallemi \inst{1,2,3} \orcidID{0000-0002-4489-9260}\and
 Mallesh M Pai \inst{4,5}\orcidID{0000-0001-9989-6676} \and
 Dan Robinson \inst{2,5}\orcidID{0009-0002-2702-0055}}
\authorrunning{Moallemi et al.}
%
\institute{Columbia University, USA \and
Paradigm, USA \and
Uniswap Labs, USA \and
Rice University, Houston, USA \and
Tempo Labs, San Francisco, USA}

}

\maketitle

\begin{abstract}
  In financial applications, latency advantages---the ability to make decisions later than others,
  even without the ability to see what others have done---can provide individual participants with
  an edge by allowing them to gather additional relevant information. For example, a trader who is
  able to act even milliseconds after another trader may receive information about changing prices
  on other exchanges that lets them make a profit at the expense of the latter. To better
  understand the economics of latency advantages, we consider a common-value auction with a
  reserve price in which some bidders may have more information about the value of the item than
  others, e.g., by bidding later.  We provide a characterization of the equilibrium strategies,
  and study the welfare and auctioneer revenue implications of the last-mover advantage. We show
  that the auction does not degenerate completely and that the seller is still able to capture
  some value. We study comparative statics of the equilibrium under different assumptions about
  the nature of the latency advantage. Under the assumptions of the Black-Scholes model, we derive
  formulas for the last mover's expected profit, as well as for the sensitivity of that profit to
  their timing advantage. We apply our results to the design of blockchain protocols that aim to
  run auctions for financial assets on-chain, where incentives to increase timing advantages can
  put pressure on the decentralization of the system.
\end{abstract}



\section{Introduction}\label{sec:intro}

In financial applications, latency advantages---the ability to make decisions later than others, even without the ability to see what others have done---can give individual participants an
edge. In particular, it allows the late mover the chance to receive and act on additional
information. For example, a trader who is able to act even milliseconds after another trader may
receive information about changing prices on other exchanges that lets them make a profit at the
expense of the latter \citep[see, for example, the discussion of stale order sniping in][]{budish}. Such advantages may deter other participants from engaging in the auction,
reducing the efficiency of markets. Nevertheless, the economics and comparative statics of
latency advantages remain unexplored. This paper provides a model of such latency
advantages in the context of auctions, and studies the economics and comparative statics of such
latency advantages.

We consider the case of an auction for a volatile asset with a reserve price. Bidders have a
common value for the asset.  To fix ideas, think of the asset being auctioned as a fixed quantity
of a stock or token, with the realized value for the winner being the prevailing market price at
the time the auction settles. Bids are received over time, and the market price of the underlying asset is observed by bidders as they bid. Later movers therefore have more information about the likely price of the asset at the time of settlement than earlier movers. We consider the case where bidders do not observe each others' bids at the time of bidding.\footnote{If later movers observe the bids of the earlier movers prior to their own bid, the auction degenerates due to an extreme lemons problem: the last mover can always outbid the earlier movers when it is profitable to do so.} At settlement, the highest bid wins the auction and pays their bid, their net profit being the difference between the prevailing price of the asset and their bid.

Our main contributions%
\ifthenelse{\boolean{camerareadyversion}}{%
  \footnote{Proofs can be found in the appendix to this paper, available at
  \url{https://arxiv.org/abs/2504.02077}.}
}{}%
can be summarized as follows. We characterize the equilibrium strategy in
this setting. While our results generalize to many bidders, for simplicity, consider for now the case of two bidders, Alice and Bob, where Bob moves after Alice. Alice has a mixed strategy, where she ``draws'' a random value from the known distribution of values and bids the conditional expected value conditioned on it being lower than her random draw. Bob, who sees the true realized value but not Alice's bid, bids the larger of the conditional expectation conditioned on being lower than the true realized value, or the reserve price, whenever the latter is profitable. When the reserve price is 0, this reduces to the solution given in \citet{engelbrecht}.

We show that the profit of the last mover (and equivalently, the loss of the seller) can be
expressed as a portfolio of certain types of options on the underlying asset. We compare this to
the profit of a single buyer (which is just the value of a call \newedit{option, i.e., an option
  to buy the asset at a fixed price}). This allows us to perform
comparative statics on the profit of the last mover, as functions of the reserve price, the
distribution of the value, and the timing advantage. In particular, for this, we assume that the
distribution of the value is governed by a geometric Brownian motion, so that we can do
comparative statics with respect to time and volatility, and compare with standard results from
the literature on option pricing (e.g., the Black-Scholes formula).  We quantify the ``timing
pressure'' caused by the last-mover advantage. Specifically, we measure the marginal profit that
the last mover gains from increasing their timing advantage, which can also be interpreted as the
``theta'' (time sensitivity) of this option.

\subsection{Applications}
Understanding the economics of latency advantages is important for a number of financial applications, but the main application we consider in this paper is to the design of (applications on) blockchain protocols.

In particular, many kinds of applications on a blockchain can be implemented as an auction, so the ability to run on-chain auctions would let applications capture value from financial transactions. Examples include:
\begin{itemize}
\item Decentralized exchange routers that auction off a user's market order to get them the best possible price,
\item NFT launches or marketplaces that auction off an NFT,
\item Automated Market Makers (AMMs) that auction off the right to trade first in order to capture
  loss-vs-rebalancing \citep[see][]{milionis2022automated} for their liquidity providers, 
\item Orderbooks that auction off the right to trade against each limit order,
\item Lending protocols that auction off collateral during liquidation,
\item Wallets that auction off the right to backrun arbitrary user interactions, for the benefit of that user
\end{itemize}

To allow on-chain auctions to work for the benefit of applications and users, rather than leaking significant amounts of value to a monopolist proposer/ block builder, the platform needs to guarantee certain properties. Enforcing some of these properties trustlessly is difficult.

\paraheader{Transaction ordering.} For auctions that settle within a single block, we need some
rules about transaction ordering. For example, ordering transactions in descending order of
priority fees (combined with the other properties here) allows some kinds of efficient same-block
auctions \citep[see, e.g.,][]{priority}. Alternatively, applications can have auctions settle in a subsequent block. Ordering rules are often easy to enforce at the protocol level since they can be part of the validity condition for blocks.

\paraheader{Censorship resistance.} A key desideratum for on-chain auctions is censorship
resistance. Without it,  the miner will simply censor all competing bids
\citep[see, e.g.,][]{fox2023censorshipresistanceonchainauctions}. One solution, as suggested in that paper, is
to have multiple concurrent proposers. Another is to run the auction over multiple blocks, though
this exacerbates the last mover advantage and creates other issues such as latency for the users.

\paraheader{Pre-transaction privacy.} Without privacy of transactions, a bidder can penny everyone
else's bid (i.e. outbid by a very small amount whenever it is profitable to do so). When combined
with even a tiny informational advantage, this can cause an auction to unravel, as shown in
\citet{pai2023structuraladvantagesintegratedbuilders}. Potential solutions to achieve this
on-chain use some form of encryption, including commit-reveal schemes, threshold encryption
\citep[such as][]{threshold}, or timelock encryption \citep[such as][]{timelock}.

\paraheader{Last mover advantage.} Finally, we need to minimize last mover advantage, since otherwise the last mover may have superior information (such as information about prices on centralized exchanges). This potentially discourages other bidders from bidding their full valuation.

\bigskip

That last one---the last mover advantage---is the most difficult of these advantages to mitigate in a decentralized protocol. This is because delays by a single party in a trustless setting could simply be the result of network delays etc. Further, some protocols---even ones that attempt to guarantee censorship resistance---have explicit last movers, such as Multiplicity \citep{multiplicity2023} and FOCIL \citep{thiery2024focil}. Leaderless auctions \citep{leaderless} is an attempt at a protocol with no last mover, but it requires $\frac{2}{3}$ honest participants and a network synchrony assumption.

The potential profit from being the last mover can incentivize participants to invest money and
effort in acquiring or increasing a timing advantage, potentially undermining the protocol. The
incentive for trading firms to invest in an ``arms race'' around low-latency communications and
computational infrastructure has been well-studied in the context of traditional finance,
including in \citet{budish}. In the decentralized context, this can lead to ``timing games'' \citep{schwarzschilling2023timemoneystrategictiming} where network participants send messages at the very last moment allowed by the protocol, which can be destabilizing. Further, it may create incentives for co-location, consolidation, collusion, and corruption among validators, all of which could reduce the effective decentralization of the protocol. It follows then that the greater the profit from increasing this advantage, the greater the ``timing pressure'' on the protocol.

\subsection{Literature Review}
In this section we briefly review the related auction literature (the applications to blockchains have already been discussed inline above).
There is a long literature on auctions with common values and asymmetric information, dating back to the PhD thesis of \citet{ortega1968models}. The most closely related work is probably \citet{engelbrecht}, who characterize the equilibrium in the case of no reserve price ($L=0$). Their analysis does not extend to the case with a positive reserve price $L > 0$ for a fundamental reason: the structure of the equilibrium changes qualitatively when a reserve price is introduced.

In the case without a reserve price ($L=0$), they show that the equilibrium has a particularly clean structure: Bob bids the conditional expectation $\mathbb{E}[\tilde{v}| \tilde{v} < v]$ for all $v > 0$, and Alice randomizes by drawing a value $v'$ and bidding $\mathbb{E}[\tilde{v}| \tilde{v} < v']$. This structure works because there is no lower bound on bids, and the tie-breaking considerations are simpler.

However, when $L > 0$, the equilibrium becomes more complex because Bob's bidding strategy must account for the reserve price constraint. Specifically, when $\mathbb{E}[\tilde{v}| \tilde{v} < v] < L$, Bob cannot bid his unconstrained optimal amount (which would be below the reserve price). Instead, Bob must either bid $L$ (if it's profitable, i.e., if $v > L$) or bid $0$ (if $v \leq L$). This creates a discontinuity (and an atom) in Bob's bidding function at the threshold value $\underline v$ defined by $L = \mathbb{E}[\tilde{v}| \tilde{v} < \underline v]$.

This discontinuity fundamentally changes the equilibrium analysis. In the framework of \citep{engelbrecht}, the smooth structure of the bidding functions allows for a relatively straightforward characterization. With a reserve price, the equilibrium requires solving for the threshold $\underline v$ that equates the reserve price with the conditional expectation, and then handling the different cases (above threshold, between $L$ and threshold, and below $L$) separately. This additional complexity is part of what we handle here. Further, our interpretation of the last mover's value in terms of options is novel and insightful in its own right, and Further, leads to novel results about the value of latency advantages and the resulting comparative statics. 

\citet{dubra2006correction} provides a complete proof that the equilibrium in this case is unique. Of course the seminal work on common value auctions is \citet{milgrom1982theory}, but it allows for all buyers to have (ex-ante symmetric) private information. \citet{hausch1987asymmetric} considers the case of asymmetrically informed bidders in a common-value auction, but restricts attention to a finite type space for tractability. Further, they consider an inherently static setting so there is no notion of a last mover or comparative statics with respect to the timing advantage.



\section{Model}\label{sec:model}

The core of our analysis is based around a standard 2-period, 2-bidder common-value auction. We describe the model in detail below, and then explain how we can either use/ reinterpret this analysis to capture several generalizations.
\begin{enumerate}
    \item There is a single good for sale.
    \item In Period 1, Alice submits a bid for the good. At this point it is only known that the
      value of the good is $v \sim F_v$, where the distribution $F_v$ has density $f_v$ with
      support $\mathbb{R}_+$. We assume that this distribution has a finite
      expectation.\footnote{We'll maintain the convention that $\tilde{v}$ is the underlying
        random variable, $v$ is a realized value. }
    \item In Period 2, nature draws $v$. Bob observes $v$ and submits a bid $\beta(v)$. However, Bob does not observe Alice's bid.
    \item There is a limit price of $L$. The winner is the highest bidder, and pays their bid, unless their bid is less than $L$, in which case the good stays unsold.
\end{enumerate}
For simplicity, we employ the following tie-breaking rule: if Alice and Bob tie, then Bob wins. This avoids the issue of Bob needing to slightly increase his bid to avoid ties. Similarly, we assume that when Bob ties with the limit price, he wins, whereas if only Alice ties with the limit price, then the good stays unsold. 

The tie-breaking rule serves two important purposes in our analysis. First, it ensures that Bob's best response is well-defined: when Bob observes a value $v$ such that his optimal bid would tie with Alice's bid, the rule that Bob wins eliminates the need for Bob to bid an infinitesimally higher amount. This avoids technical complications with discontinuous best responses. Second, and more importantly, the tie-breaking rule is crucial for ensuring that Alice's equilibrium strategy yields zero expected payoff, as we show in our equilibrium characterization. We note that alternative tie-breaking rules (e.g., random tie-breaking or Alice wins ties) would not qualitatively change our results, but would complicate the analysis. A more detailed discussion of the tie-breaking rule and its role in the equilibrium appears in Section \ref{sec:results}.

We look at standard Perfect Bayesian Nash Equilibria of this auction as a function of the primitives of the model, i.e., the distribution of $v$ (i.e, $F_v, f_v$), and the limit price $L$.

We then consider the following extensions:
\begin{enumerate}
    \item The case where there are more than two bidders follows easily and is described in \Cref{corollary:multiple_agents}.
    \item We consider the case of uncertain timing advantages where there may be only one bidder, i.e., Bob may not exist and Alice may be last mover with some probability $\alpha \in (0,1)$, and, with the remaining probability $1-\alpha$, Bob is the last mover. This is described in Section \ref{sec:uncertain-timing-advantages}.
    \item While we describe this as a two-period model, it can be straightforwardly embedded in a
      continuous-time setting where Alice moves first (at time $0$) and Bob moves at time $T > 0$,
      and the price of the object evolves over time, with the final realized value of the object
      to the winner being at settlement time $\tau > T$. This is described in
      \Cref{sec:cts}.
    \item This can also be extended to the setting where both Bob's timing advantage $T$ over Alice
      and the final settlement time $\tau$ are both stochastic.
\end{enumerate}

In what follows, after calculating the equilibrium of the model, we then show comparative statics on various quantities of interest (seller revenue, bidder surplus, etc.) on the primitives of the model and various timing parameters (e.g., probability of being last mover $\alpha$, magnitude of Bob's timing advantage $t$, etc.).

\subsection{Application to On-Chain Auctions}\label{sec:on-chain-auction}

We now describe how this abstract model applies to on-chain auctions. Suppose that there is a single order (for example, on a variant of UniswapX), corresponding to a user looking to sell a quantity $q$ normalized to $1$ of token $A$ at minimum price of $L$ (denominated in token $B$). For example this could be a limit order to sell $1$ unit of ETH at a minimum price of $L$ (denominated in USD). Alice and Bob are the two bidders in this auction, they can each only submit a single bid, corresponding to the price at which they offer to fill the order. For simplicity suppose there is only one such order, and partial fills are not allowed, i.e. the winner is the bidder with the highest bid, and pays their bid price.

The current price of token $A$ on a reference centralized exchange at the time of Alice's bid is
$p_0$, hence this represents Alice's valuation of the good at time $0$. Price evolves as a
geometric Brownian motion with zero drift  and
volatility $\sigma$.\footnote{Because of the short time scales (e.g., less
than a minute) of such an auction in practice, the drift of the price process does not
significantly impact the results, and hence we can assume zero drift. This could alternatively
be justified by valuing the auction under the risk neutral distribution (where the drift would
equal the risk-free rate), since the risk-free rate over such time scales is also zero.} When Bob bids at time $t$, the price is $p_t$. In particular, Bob knows the prevailing price $p_t$ at
the time of his bid, whereas Alice only knows the initial price $p_0$, and at the time of her bid does not know what the price of the object will be when it is sold.

One possible instantiation of this idea is a sealed-bid auction on chain that takes place over two blocks, where the two proposers (Alice as the proposer of block $N$, and Bob as the proposer of block $N+1$) are independent. Alice can place a commitment to her bid in block $N$, and Bob can place one in block $N+1$. Both must subsequently reveal their bid or sacrifice their deposit.\footnote{While this is out of scope for this paper, we also assume there is some fixed deposit amount that is always sufficient to discourage strategic withholding of bids, to prevent the issues discussed in \citet{schlegel2022onchain}.} The second proposer does not know the bid of the first proposer, but gets to submit their bid at a later point in time, which may give them additional information about the price of the token.

Another possible instantiation is a limit order that uses MEV taxes (as described in \citet{priority}) on a chain that uses some cryptographic and consensus mechanisms (such as some variation on multiple concurrent proposers, described in \citet{fox2023censorshipresistanceonchainauctions} and \citet{multiplicity2023}) to guarantee censorship resistance and privacy of transactions, but where there is one actor that has the ability to add transactions later than anyone else.

Other than grounding our model, this instantiation  allows us to consider comparative statics
exercises on the \emph{timing pressure} or 
Bob's incentive to wait, i.e., to bid at time $T+\Delta t$ with additional information about the
price, as opposed to just bidding at $T$.



\section{Results}\label{sec:results}

Our main result characterizes the equilibrium bidding strategies for Alice and Bob, and compares the expected payoffs across different parameter values. In particular, we are able to fully characterize the equilibrium strategies for both bidders as functions of the primitives of the model, and then connect the equilibrium payoff of Bob to the value of certain derivatives we describe below.

\subsection{Equilibrium in the Baseline Model}

\begin{theorem}\label{theorem:equilibrium}
    The equilibrium bidding strategies for Alice and Bob are:
    \begin{itemize}
        \item Bob bids $\beta_L(v)$ as a function of the observed value $v$ of the item:
          \begin{align}\label{eq:bob_bid}
         \beta_L(v) =   \begin{cases}
                \E[\tilde v| \tilde v < v] & \text{if $v \geq \underline v$}, \\
                L & \text{if $L < v < \underline v$}, \\
                0 & \text{if $v \leq L$}
            \end{cases}
        \end{align}
      where $\underline v$ is the value such that
    \begin{align}\label{eq:L_def}
        L= \mathbb{E}[\tilde{v}| \tilde{v} < \underline v].
    \end{align}
     \item Alice's (mixed) bidding strategy can be described as follows: she first randomly draws a value $v'$ from the known distribution $F_v$. She then submits a bid $b_a$ equal to $\beta_L(v')$.
    \end{itemize}
\end{theorem}

Before describing the proof of this result, a few comments are in order regarding the existence of a solution $\underline v$ to \eqref{eq:L_def} and the tie-breaking rule.

\paraheader{Existence of solution to \eqref{eq:L_def}.} The function $g(v) = \mathbb{E}[\tilde{v}| \tilde{v} < v]$ is well-defined for all $v$ in the support of $F_v$ and is continuous and strictly increasing in $v$ (since the conditional expectation increases as we condition on larger values). Moreover, by the law of total expectation, we have $\lim_{v \to \inf \text{supp}(F_v)} g(v) = \inf \text{supp}(F_v)$ (when the infimum is finite) and $\lim_{v \to \sup \text{supp}(F_v)} g(v) = \mathbb{E}[\tilde{v}]$ (when the supremum is finite). 

Therefore, a solution $\underline v$ to \eqref{eq:L_def} exists if and only if $L$ lies in the range of $g$, which is the interval $(\inf \text{supp}(F_v), \mathbb{E}[\tilde{v}])$. In particular, if $\inf \text{supp}{F_v}<L < \mathbb{E}[\tilde{v}]$, then there exists a unique solution $\underline v$ such that $L = \mathbb{E}[\tilde{v}| \tilde{v} < \underline v]$.

The economic intuition behind this condition is as follows. The threshold $\underline v$ represents the cutoff value such that when Bob observes $v = \underline v$, his optimal bid (the conditional expectation $\mathbb{E}[\tilde{v}| \tilde{v} < \underline v]$) exactly equals the reserve price $L$. For this to be possible, the reserve price must be between the lowest possible value in the support and the unconditional expectation $\mathbb{E}[\tilde{v}]$ (which is what Bob would bid if he observed the highest possible value) . If $L$ is too low, then even when Bob observes very low values, his optimal bid would exceed $L$, so the reserve price never binds. If $L$ is too high (above the unconditional expectation), then the reserve price is so restrictive that Alice would not bid and Bob wins at $L$ whenever his observed value is above the threshold. Our condition therefore ensures a nontrivial auction. 

\paraheader{Tie-breaking rule.} The tie-breaking rule we employ serves two important purposes. First, it ensures that Bob's best response is well-defined: when Bob observes a value $v$ such that his optimal bid would tie with Alice's bid, which is possible with strictly positive probability due to the presence of atoms in the bidding functions, the rule that Bob wins eliminates the need for Bob to bid an infinitesimally higher amount. This avoids technical complications.

Second, and more importantly, the tie-breaking rule is crucial for ensuring that Alice's equilibrium strategy yields zero expected payoff. Specifically, when Alice bids at or below the limit price $L$, she ties with the limit price and loses (the good stays unsold). This prevents Alice from having a safe strategy (bid $L$ and win with positive probability) that yields a positive expected payoff.

We note that alternative tie-breaking rules (e.g., random tie-breaking or Alice wins ties) would not qualitatively change our results, but would complicate the analysis. In particular, if ties were broken randomly, Bob would need to account for the probability of winning conditional on tying, which would shift his bidding function but preserve the essential structure of the equilibrium. The key insight that Alice earns zero expected profit and Bob's informational advantage translates into positive expected profit would remain unchanged.
\bigskip

If there are multiple agents bidding in period $0$, i.e., prior to the value being known, the equilibrium bidding strategy for each agent would be any strategy such that the distribution of the maximum bid is the same as that of Alice above.
\begin{corollary}\label{corollary:multiple_agents}
  Consider the case where there are $n$ bidders, $n-1$ of which bid in period $1$ and the remaining bidder bids in period $2$. Then the equilibrium bidding strategy for each agent in period $1$ is any tuple of strategies such that the distribution of the maximum bid over all $n-1$ period 1 bidders is the same as that of Alice above, while the period 2 bidder's strategy is the same as that of Bob above.
\end{corollary}
\begin{proof}
    Note that since the distribution of the max of the $n-1$ ``Alice'' bidders is the same as in the case of two bidders, ``Bob's'' best response stays the same. Further, as we show below in the Proof of Theorem \ref{theorem:equilibrium}, Alice makes a profit of 0 over all bids in her support, so any mixed strategy profile for all $n-1$ period 1 bidders that retains the same distribution of max also gives each bidder a payoff of 0 for every bid in the support, and is therefore a best-response.  
\end{proof}
Note that the case with multiple agents bidding in period $2$ is uninteresting since each of them will bid exactly the observed value $v$, while all agents bidding in period $1$ will now bid $0$.

\subsubsection{Proof of Theorem \ref{theorem:equilibrium}}

In this section we prove Theorem \ref{theorem:equilibrium}. The proof is instructive but can be skipped without loss of continuity on a first reading.

The first Lemma essentially tells us that there are no pure strategy equilibria:
\begin{lemma}\label{lemma:no_pure_strategies}
    There is no equilibrium in this model in pure strategies for Alice.
 \end{lemma}
 \begin{proof}
     Suppose Alice has an equilibrium bid in pure strategies. Then Bob knows Alice's bid $b_a >0$. Bob can then bid $\max (b_a,L)$ whenever $v > \max (b_a,L)$ and $0$ otherwise. In this case Alice's expected payoff is negative, and simply bidding 0 would be strictly better. However, Alice's equilibrium bid cannot be $0$ either, since Bob will optimally respond with a bid of $L$ whenever $v > \max (b_a,L)$, and Alice's best response to Bob's bid of $L$ is not $0$ given our assumption that $\mathbb{E}_{\tilde{v}}[\tilde{v}]> L$.
 \end{proof}

 Next, we argue that in any equilibrium in this model, Alice makes a payoff of $0$.

 \begin{lemma}\label{lemma:alice_zero_payoff}
     In any equilibrium, Alice makes 0 expected payoff for all bids that she submits.
 \end{lemma}
 \begin{proof}
         Lemma \ref{lemma:no_pure_strategies} tells us that the equilibrium must involve Alice randomizing.

         Next, note that usual equilibrium arguments tell us that Alice must be indifferent over all bids in the support of her randomization.

         Let $B_A \subseteq \mathbb{R}_+$ be the support of bids Alice makes. Note that $\inf B_A \leq L$, i.e. Alice makes bids that are guaranteed to not win sometimes.

         To see why, suppose not, i.e., suppose $\inf B_A = \underline{b}> L$. Therefore Bob knows that any bid Alice makes is at least as large as $\underline{b}$. It follows that Bob should either not bid (or bid $0$) whenever value $v_t \leq \underline{b}$, or bid more than $\underline{b}$ whenever value $v_t > \underline{b}$. Therefore bids arbitrarily close to $\underline{b}$ in $B_A$ by Alice must have strictly negative expected value, a contradiction.

         However note that she makes $0$ expected payoff from any bid $b$ such that $b = \inf B_A$ (due to the tie-breaking rule).%
         \footnote{Note that a different tie-breaking rule would not change our results qualitatively, but would make defining Bob's best response more complicated.} Since Alice is indifferent over all bids in the support of her randomization, it must be that she makes $0$ expected payoff from any bid in $B_A$.
 \end{proof}

With these preliminaries out of the way, we now show that the strategies described above for Alice and Bob form an equilibrium. Uniqueness follows from the result of \citet{dubra2006correction}.

\begin{lemma}\label{lemma:alice_bids}
  Suppose Bob's bid as a function of the value $v$ is given by a non-decreasing function $\beta_L(v)$.  Therefore, we must have, for any bid of Alice  $b_A > L$,
  \[
    b_A = \E\left[  v | v < \beta_L^{-1}(b_A) \right].
  \]
\end{lemma}

Let $\underline{v}$ be such that $L = \E[v|v\leq
  \underline{v}]$. In light of this Lemma, by a straightforward change of variables, we can
  observe that for any $ v > \underline{v}$,
  $\beta_L(v) =\E[\tilde v| \tilde v \leq v]$.
  If $vleq L$, then trivially Bob will not want to win the auction (or will be indifferent), so
  $\beta_L(v) = 0$ if $v \leq L$.
 For $v \in (L,\underline{v})$,  it's still profitable to bid $L$, so let's conjecture that $\beta_L(v) = L$. Our conjectured bidding strategy therefore is
\[
    \beta_L(v) =
    \begin{cases}
      \E[\tilde v| \tilde v < v] & \text{if $v \geq \underline v$}, \\
      L & \text{if $L < v < \underline v$}, \\
      0 & \text{if $v \leq L$}.
    \end{cases}
\]

Now let's recall our proposed equilibrium strategy for Alice: Alice randomly draws a value $\tilde{v}$ according to $F_v$ and bids $\mathbb{E}[v| v < \tilde{v}]$ whenever this is larger than $L$ and $0$ otherwise. Note that given this strategy by definition, Alice gets an expected profit of $0$ regardless of her bid whenever Bob's bid is according to $\beta_L(\cdot)$, as was shown in Lemma \ref{lemma:alice_bids}.

To conclude the proof therefore we have the following Lemma:
\begin{lemma}\label{lemma:bob_bids}
  Bob's best response to Alice's strategy as conjectured in the statement of the theorem is given by $\beta_L(\cdot)$ as conjectured above in \eqref{eq:bob_bid}.
\end{lemma}

\subsection{Equilibrium payoffs}
As we already argued above, Alice gets an expected payoff of $0$ from any bid that she submits. In some ways, this is intuitive, since she has no information about the value and therefore cannot gain any information rents. Further, given the description of the equilibrium bidding strategies, both Alice and Bob each win the good with ex-ante probability $1/2$--- i.e., even though Bob has an informational advantage, this does not translate into a higher probability of winning. On the other hand the informational advantage does translate into a higher expected payoff for Bob:
\begin{theorem}\label{theorem:bob_profit_0}
    For the case that $L=0$, Bob's expected profit $\Pi_B$ is given by:
    \begin{equation}\label{eq:bob-l-zero}
      \Pi_B = \E[\max(v_1 - v_2,0)],
    \end{equation}
    where $v_1, v_2$ are i.i.d.\ draws from $F_v$.
\end{theorem}

\newedit{ In the following corollary, we observe that \Cref{theorem:bob_profit_0} can be
  interpreted as the price of an \emph{exchange option}. In particular, imagine two assets with
  values $v_1$ and $v_2$ independently distributed according to $F_v$. An option to exchange the
  second asset for the first asset has payoff $\max(v_1 - v_2, 0)$, which corresponds
  to Bob's
  expected profit in \eqref{eq:bob-l-zero}.}

\begin{corollary}\label{corr:exchange_option}
  Bob's expected payoff $\Pi_B$ when $L=0$ in \eqref{eq:bob-l-zero} is the payoff of an ``exchange
  option'' \citep{margrabe1978value}, i.e., the option to exchange one asset for another, in the
  case where both assets have i.i.d.\ values from the same underlying distribution $F_v$. This
  option is worth more than either a call or put option with a strike price that is at-the-money.
\newedit{Such call option would be an option to buy an asset with value $v$ at price $\E[v]$,
  i.e., a payoff of $\max(v - \mathbb{E}[v], 0)$. Similarly, a put option would be an option to sell
  an asset with value $v$ for a price $\E[v]$, i.e., a payoff of $\max(\mathbb{E}[v] - v, 0)$.}
  Then,
  \[
    \Pi_B \geq  \E\left[  \max(v - \mathbb{E}[v], 0) \right],\quad
    \Pi_B \geq  \E\left[  \max(\mathbb{E}[v] - v, 0) \right].\quad
  \]
\end{corollary}
\begin{proof}
  Note that, by the tower property of expectation and Jensen's inequality,
  \[
    \begin{split}
      \Pi_B & = \mathbb{E}[\max(v_1 - v_2, 0)]
      =  \mathbb{E}\left[ \mathbb{E}[ \max(v_1 - v_2, 0) | v_2 ] \right] \\
      & \geq \mathbb{E}\left[  \max(v_1 - \mathbb{E}[v_2], 0) \right]
      = \mathbb{E}\left[  \max(v - \mathbb{E}[v], 0) \right], \\
    \end{split}
  \]
  so the exchange option is worth more than an at-the-money call. The same argument can establish
  that it's also worth more than an at-the-money put.
\end{proof}

\begin{theorem}\label{theorem:bob_profit_L}
  In the general case in which $L>0$. In this case, Bob's expected profit $\Pi_B$ is given by
  \[
    \begin{split}
      \Pi_B
      & = \E\left[\max\left(v_1 \I{v_1 \geq \underline{v}} - v_2 \I{v_2 \geq \underline{v}},
        0\right)\right]
        + \E\left[(v_1 - L)\I{L < v_1 < \underline{v}} \I{v_2 \geq \underline{v}}\right]
        - \E[ v_1 \I{v_2 > \underline{v}}]
      \\
      & = \E\left[\max\left(v_1 \I{v_1 \geq \underline{v}} - v_2 \I{v_2 \geq \underline{v}},
        0\right)\right]
        + F_v(\underline{v}) \E\left[(v - L)\I{L < v < \underline{v}}\right]
        - (1-F_v(\underline{v})) \E[ v],
    \end{split}
  \]
  where $\underline v$ is the value such that $L= \mathbb{E}[\tilde{v}| \tilde{v} < \underline v]$.
\end{theorem}

Each of the terms in the above theorem has a natural economic interpretation:
\newedit{
The first term is still an exchange option. However the underlying assets on which the
  exchange option is written are asset-or-nothing options on two i.i.d.\ assets with strike price
  $\underline{v}$. An asset-or-nothing option, in turn, has a payoff equal to the value of the
  asset but only if the value exceeds the strike price. In other words, assets with payoffs
  $v_1 \I{v_1 \geq \underline{v}}$ and $v_2 \I{v_2 \geq \underline{v}}$ are the constituents of
  the exchange option.
The second term corresponds to being long $F_v(\underline{v})$ units of a type of range
  option. This range option is a call option with strike price $L$ (i.e., an option to buy at the
  price $L$) that only pays out if, in addition, the value  is less than $\underline{v}$.
The final term can be interpreted as being short $\E[v]$ units of a binary
  call with strike price $\underline{v}$. Such a binary call gives a payoff of $1$, but only if
  the final value of the asset exceeds $\underline{v}$.
}
Note that when $L=0$, the latter terms vanish and the first term corresponds to the payoff of an
exchange option as described in Theorem \ref{theorem:bob_profit_0}.

The revenue of the auctioneer is calculated in the ex-ante sense, i.e., the expected revenue
before the auction is held, expectation taken over the distribution of the value of the item. It
is straightforward given the equilibrium bidding strategies for Alice and Bob to establish that:
\begin{theorem}\label{theorem:revenue}
  The auctioneer's expected revenue is given by
  \[
    \text{Revenue} = \E[\max\{\beta_L(v_1), \beta_L(v_2)\}],
  \]
  where $v_1,v_2$ are i.i.d.\ draws from $F_v$.
\end{theorem}
Again this has an option interpretation: the auctioneer's revenue is equivalent to \newedit{an
  option to choose one of two assets, each of which is distributed i.i.d. as $\beta_L(\tilde{v})$, where $\tilde{v}$ is a random variable with distribution $F_v$.}


\subsection{Extension: Uncertain Timing Advantages}\label{sec:uncertain-timing-advantages}

We now consider the case where Bob may not be able to act in time with some probability $\alpha
\in (0,1)$, in which case Alice is the last mover. For simplicity, in this extension we assume
that there is no reserve price, i.e., $L=0$.

Note that the previous equilibrium now no longer persists: in the previous equilibrium, Alice made
exactly $0$ and was indifferent over all her bids, while now she can guarantee herself a payoff of
$\alpha \mathbb{E}[v]$ by simply bidding the reserve price of $L$ and pocketing the difference
whenever Bob does not appear. Nevertheless, repeating the logic of the previous theorem delivers a
characterization of equilibrium.

\begin{theorem}\label{theorem:uncertain-timing}
    In the case where Bob may not be able to act in time with probability $\alpha \in (0,1)$, the equilibrium is given by the following:
    \begin{itemize}
        \item Bob (when he is able to act) bids $\beta(v)$ as a function of the observed value $v$:
        \begin{align*}
            \beta(v) = \begin{cases}
                0 & \text{if } v \leq \bar{v}_\alpha, \\
                \E\left[ \tilde{v} | \tilde{v} < v \right] - \frac{\alpha \mathbb{E}[\tilde{v}]}{P(\tilde{v} < v)} & \text{if } v > \bar{v}_\alpha,
            \end{cases}
        \end{align*}
        where $\bar{v}_\alpha$ satisfies
        $\alpha \mathbb{E}[v] = \E\left[  v | v < \bar{v}_\alpha \right] \mathbb{P}(v < \bar{v}_\alpha).$
    \item Alice employs a mixed strategy, randomly drawing value $v'$ from the distribution $F_v$ and bidding $\beta(v')$ as a function of the drawn value $v'$.
    \end{itemize}
\end{theorem}

Therefore in this case, by observation, Alice makes a positive profit in equilibrium, but the structure of the equilibrium remains similar to the case with no uncertainty. Further the comparative statics remain somewhat similar: Alice's rents are proportional to $\alpha$, the probability with which she is the last-mover. Bob benefits from the additional shading in Alice's bid relative to the case with no uncertainty, while the seller revenue is lower even conditional on the event that Bob is able to act.


\section{Continuous-Time Setting and Comparative Statics}\label{sec:cts}

In this section, we consider an instantiation of our auction model in a continuous-time
setting, in the context of the on-chain auction of Section \ref{sec:on-chain-auction}.  Here, we
imagine the auction is selling a single unit of token $A$, whose price evolves according to a
geometric Brownian motion with volatility $\sigma$, i.e., it satisfies the stochastic differential
equation (SDE)   $dp_t/p_t = \sigma\, dB_t$,
where $B_t$ is a standard Brownian motion. This is consistent with the standard assumptions of the
Black-Scholes model for option pricing.

We imagine Alice bids at time $0$, Bob bids at time $T > 0$, and the final settlement occurs at
time $\tau > T$. All agents are risk neutral (and, given the short time horizon, we assume a zero risk-free rate), so that, at
any time $t$, the asset is valued according to $\E[p_\tau|p_t] = p_t$.

Then, the value of the auction is determined by the price at
time $t=T$ when Bob bids, i.e., $v = \E[p_\tau|p_T] = p_T$. Solving the SDE for $p_T$,
\newedit{it is easy to see that $p_T$ is lognormally distributed, so} this is equivalent to Bob
observing value $v = p_T = p_0 e^{-\sigma^2 T/2 + \sigma \sqrt{T} Z}$,
where $p_0$ is the initial price (observed when Alice bids), and $Z$ is a zero mean, unit variance
Gaussian random variable.

Without loss of generality, we will fix the parameters $\sigma=1$ (which is equivalent to
normalizing the units of time) and $p_0=1$ (which is equivalent to normalizing the units of price).
Then, we
can calculate the value of Bob's profit and auctioneer's revenue as a function of Bob's latency
advantage $T$. For simplicity, we focus on the case $L=0$ (no reserve price), see
\Cref{fig:payoff_revenue}.\todo{typo here}

\begin{figure}[t]
  \centering
  \includegraphics[width=0.70\textwidth]{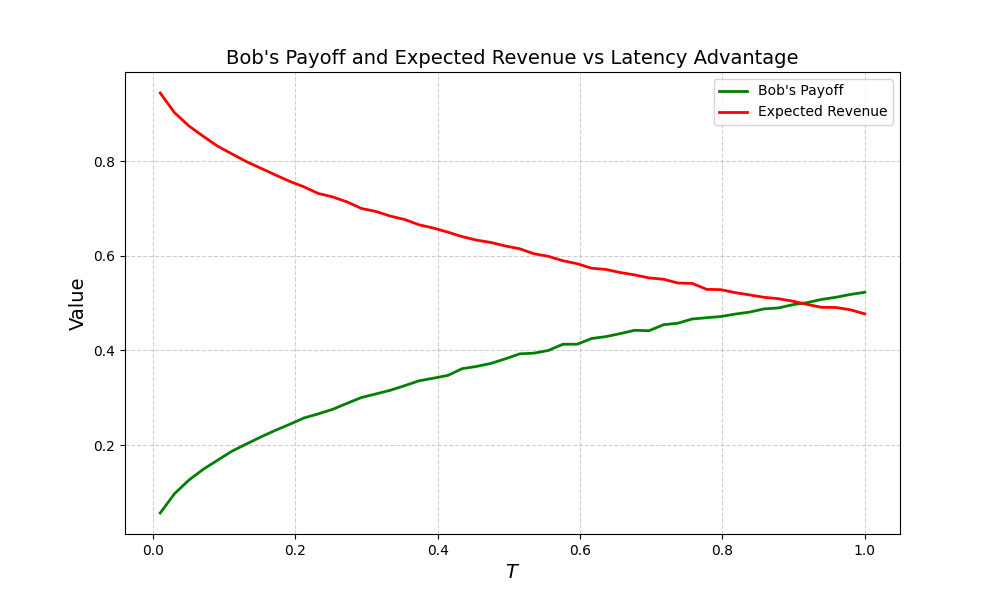}
  \caption{Bob's payoff and auctioneer's expected revenue as a function of latency advantage
    $T$. As Bob's latency advantage increases, his payoff grows while the auctioneer's revenue
    decreases.
    \label{fig:payoff_revenue}}
\end{figure}

Note that since $p_0=1$ the total surplus is 1 regardless of Bob's latency advantage. With Alice
guaranteed to make 0 expected profit, the total surplus is the auctioneer's revenue plus Bob's
payoff, i.e., this is exactly a constant sum game. As Bob's latency advantage increases, his
payoff grows while the auctioneer's revenue decreases.



\section{Value of Timing Games}\label{sec:timing-games}

In the setting of \Cref{sec:cts}, the distribution of Bob's value $v = p_T$ has clear
dependence on the time horizon $T$ at which Bob bids. We wish to quantify Bob's incentive for
\emph{timing games}: in general, Bob will profit more by extending $T$. We will compare this
timing pressure in two settings:
\begin{enumerate}
\item The competitive auction setting of our model from Section \ref{sec:model}, where Bob bids against
  Alice.
\item A monopolist setting, where Bob bids in isolation. Here, after observing $v$, Bob
  will simply bid the limit price $L$ if $v > L$, and will otherwise bid zero. Hence, Bob's
  expected profit
  in the monopolist setting is
  \[
    \Pi_M(T) \triangleq \E[ \max(v - L, 0)].
  \]
  This is the value of a call option on $v$, with strike price $L$.
\end{enumerate}

\subsection{Assuming $L=0$.}

In the competitive auction setting with no reserve price, i.e., $L=0$, denote Bob's profit by $\Pi_B(T)$, highlighting the dependence
on $T$. From \Cref{corr:exchange_option}, this is given by the value of an exchange option. Using
Margrabe's formula \citep{margrabe1978value} for the value of an exchange option under the
Black-Scholes model, we have that
\begin{align*}
  \Pi_B(T) & \triangleq \E\left[ \max(v_1 - v_2, 0) \right]
             = p_0
             \left( \Phi\left( \sigma \sqrt{T/2}  \right)
             - \Phi\left( - \sigma \sqrt{T/2}  \right)
             \right) \\
           & =
             p_0
             \left( 2 \Phi\left( \sigma \sqrt{T/2}  \right) - 1
             \right),
\end{align*}
where $\Phi(\cdot)$ is the standard normal CDF.  We can quantify the value of extending the time
interval for Bob's decision through the derivative with respect to time. \newedit{This corresponds
  to what is often called the
``theta'' of an option in the options pricing literature. The theta is the derivative of the
option value with respect to remaining time.} Here, we have that\footnote{\newedit{This expression is
  the theta of an exchange option, and is the also the
  theta of an at-the-money call or put, with half the volatility, or half the time to expiry, as
  the exchange option.}}
\[
  \Pi_B'(T)
  = p_0 \frac{\sigma}{\sqrt{2 T}} \phi\left( \sigma \sqrt{T/2}  \right) > 0,
\]
where $\phi(\cdot)$ is the standard normal PDF. In particular, the incremental value of increasing $T$ to
$T+\Delta t$, for small $\Delta t$, is
\begin{equation}\label{eq:Pi_B_T}
  \Pi_B(T+\Delta t) - \Pi_B(T) = \Pi_B'(T) \Delta t + o(\Delta t)
  = p_0 \frac{\sigma}{\sqrt{2 T}} \phi\left( \sigma \sqrt{T/2}  \right) \Delta t
  + o(\Delta t).
\end{equation}
Thus, the incremental value of a delay of $\Delta t$ is of order $\Delta t$.

On the other hand, if Bob was a monopolist, his profit would be
$\Pi_M(T) = \E\left[ v  \right] = p_0$,
so that
\begin{equation}\label{eq:Pi_M_T}
  \Pi_M(T+\Delta t) - \Pi_M(T) = 0.
\end{equation}
In this case there is zero incremental value for delays. This is because, in this case, the
monopolist payoff is equivalent to that of an option with a zero strike price. Such an option will
\emph{always} be exercised, and hence only has intrinsic value: there is no time value since
delaying cannot influence the decision to exercise.

Comparing \eqref{eq:Pi_B_T}--\eqref{eq:Pi_M_T}, we see that the timing pressure is greater
in the competitive auction. This is intuitive, since in the auction setting, increasing Bob's time
advantage over Alice increases Bob's informational advantage, and hence the value Bob can extract.

\subsection{Assuming $L>0$.}

Now, consider the more general case with reserve price $L>0$. The following corollary characterizes Bob's profit as a function of $T$.
\begin{corollary}\label{corollary:profit_L}
  Without loss of generality, assume that prices are normalized so that $p_0=1$. Then, Bob's
  profit is given by
  \begingroup
  \small
  \allowdisplaybreaks
  \begin{align}
    \Pi_B(T)
    & = 1-\Phi\left(\frac{\ln\left(\underline{v}\right)}{\sigma\sqrt{T}}-\frac{\sigma\sqrt{T}}{2}\right)\Phi\left(\frac{\ln\left(\underline{v}\right)}{\sigma\sqrt{T}}+\frac{\sigma\sqrt{T}}{2}\right)\notag\\
    & \quad-\frac{2}{\sigma\sqrt{T}}\int_{\ln\left(\underline{v}\right)}^{\infty}\phi\left(\frac{w}{\sigma\sqrt{T}}+\frac{\sigma\sqrt{T}}{2}\right)\Phi\left(\frac{w}{\sigma\sqrt{T}}-\frac{\sigma\sqrt{T}}{2}\right)\, dw\notag\\
    & \quad+\Phi\left(\frac{\ln\left(\underline{v}\right)}{\sigma\sqrt{T}}-\frac{\sigma\sqrt{T}}{2}\right)\Phi\left(\frac{\ln\left(L\right)}{\sigma\sqrt{T}}+\frac{\sigma\sqrt{T}}{2}\right)\notag\\
    & \quad-\Phi\left(\frac{\ln\left(\underline{v}\right)}{\sigma\sqrt{T}}+\frac{\sigma\sqrt{T}}{2}\right)\Phi\left(\frac{\ln\left(L\right)}{\sigma\sqrt{T}}-\frac{\sigma\sqrt{T}}{2}\right). \label{eq:bob_profit}
  \end{align}

  The sensitivity of Bob's profit with respect to the timing advantage $T$ is given by

  \begin{align}
    \Pi'_B(T)
    & = \Phi\left(\frac{\ln\left(\underline{v}\right)}{\sigma\sqrt{T}}-\frac{\sigma\sqrt{T}}{2}\right)\cdot\frac{\sigma}{\sqrt{T}}\cdot \notag\\
    & \quad \left(\frac{\phi\left(\frac{\ln\left(L\right)}{\sigma\sqrt{T}}+\frac{\sigma\sqrt{T}}{2}\right)}{2}+\frac{1}{\sigma\sqrt{T}}\left(\frac{\ln\left(\underline{v}\right)}{\sigma\sqrt{T}}-\frac{\sigma\sqrt{T}}{2}\right)\phi\left(\frac{\ln\left(\underline{v}\right)}{\sigma\sqrt{T}}+\frac{\sigma\sqrt{T}}{2}\right)\right)\notag\\
    & \quad +\frac{\sigma\cdot\underline{v}\cdot\phi\left(\frac{\ln\left(\underline{v}\right)}{\sigma\sqrt{T}}+\frac{\sigma\sqrt{T}}{2}\right)}{2\sqrt{T}\left(\underline{v}-L\right)}\left(L\cdot\Phi\left(\frac{\ln\left(L\right)}{\sigma\sqrt{T}}+\frac{\sigma\sqrt{T}}{2}\right)-\Phi\left(\frac{\ln\left(L\right)}{\sigma\sqrt{T}}-\frac{\sigma\sqrt{T}}{2}\right)\right)\notag\\ \displaybreak[2]
    & \quad +\frac{1}{2\sqrt{2}}\frac{\sigma \cdot e^{-\frac{\sigma^{2}T}{4}}}{\sqrt{2\pi}\sqrt{T}}\left(\frac{\sqrt{2}}{\sigma\sqrt{T}}\phi\left(\frac{\ln\left(\underline{v}\right)}{\left(\frac{\sigma\sqrt{T}}{\sqrt{2}}\right)}\right)+1-\Phi\left(\frac{\ln\left(\underline{v}\right)}{\left(\frac{\sigma\sqrt{T}}{\sqrt{2}}\right)}\right)\right)\notag\\
    & \quad +\frac{\sigma}{\sqrt{T}}\left(\frac{1}{\sigma^{2}T}+\frac{1}{4}\right)\int_{\ln\left(\underline{v}\right)}^{\infty}\phi\left(\frac{w}{\sigma\sqrt{T}}+\frac{\sigma\sqrt{T}}{2}\right)\Phi\left(\frac{w}{\sigma\sqrt{T}}-\frac{\sigma\sqrt{T}}{2}\right)\, dw\notag\\
    & \quad -\frac{\sigma}{\sqrt{T}}\frac{1}{\left(\sigma^{2}T\right)^{2}}\int_{\ln\left(\underline{v}\right)}^{\infty}\phi\left(\frac{w}{\sigma\sqrt{T}}+\frac{\sigma\sqrt{T}}{2}\right)\Phi\left(\frac{w}{\sigma\sqrt{T}}-\frac{\sigma\sqrt{T}}{2}\right)w^{2}\, dw. \label{eq:bob_theta}
  \end{align}
  \endgroup
\end{corollary}
\begin{proof}
  \eqref{eq:bob_profit} follows from the direct evaluation of the expression in
  \Cref{theorem:bob_profit_L} in the present context. \eqref{eq:bob_theta} follows from the
  differentiation of \eqref{eq:bob_profit}, observing the fact that $\underline{v}$ depends on $T$
  through the relationship $L= \mathbb{E}[v| v < \underline v]$.
\end{proof}

By contrast, in the monopolist setting, the time pressure is equal to the theta of a call option
with strike price $L$ (again with prices normalized so that $p_0 = 1$):

\[
  \Pi_M'(T)
  = \frac{\sigma}{2 \sqrt{T}} \phi\left( \frac{\ln\left(\frac{1}{L}\right)}{\sigma\sqrt{T}}+\frac{\sigma\sqrt{T}}{2} \right).
\]

By plotting this timing pressure as a function of $L$ (see Figure \ref{fig:theta}), we can see that for
the last mover, timing pressure appears to be higher for the last mover than for the monopolist,
but that they converge as $L$ approaches $p_0=1$.

\begin{figure}[t]
  \begin{center}
    \includegraphics[width=0.70\textwidth]{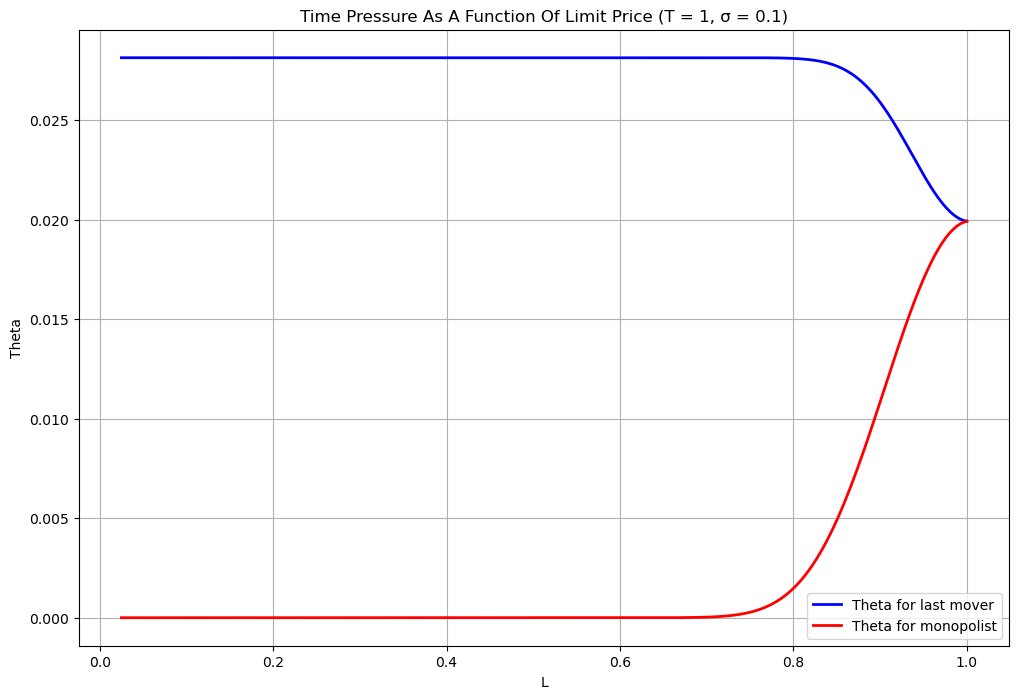}
  \end{center}
  \caption{The last mover's theta (timing pressure) as a function of the limit price, compared
    with the monopolist's theta as a function of the limit price.}\label{fig:theta}
\end{figure}





%
{\scriptsize\smallskip\noindent\textbf{\ackname.} The authors thank Max Resnick for helpful discussions and comments. The first author acknowledges support from the Briger Family Digital Finance Lab at Columbia Business School.\par}

\begingroup
\scriptsize
\sloppy
\setlength{\bibsep}{0pt plus 0.2ex}
\def\UrlFont{\rmfamily\scriptsize}
\Urlmuskip=0mu plus 1mu
\bibliography{references}
\endgroup
\ifthenelse{\boolean{camerareadyversion}}{%
}{%
  \appendix
  \section{Appendix: Proofs}

\subsection{Proof of Lemma \ref{lemma:alice_bids}}
\begin{proof}
    Suppose Bob's bid as a function of the value $v$ is $\beta_L(v)$. Assume this is non-decreasing in $v$. Define $\beta_L^{-1}$ to be the pseudo-inverse of $\beta_L$, i.e. $\beta_L^{-1}(b) = \inf \{v: \beta_L(v) \geq b\}$.

    Now, note that Alice would never bid $b_A > \mathbb{E}[v]$, since in that case her payoff would be
  \[
    \E\left[ (v - b_A) \I{b_A \geq
        \max\{\beta_L(v), L\}
      }
      \right] \leq \E[v] - b_A < 0.
   \]
  So, $\mathcal{B}_A \subset [0,\E[v]]$. This also implies Bob's bid is bounded,  $\beta_L(v) \leq \mathbb{E}[v]$, as well.
  For any bid $b_A\in\mathcal{B}_A$ with $b_A > L$, Alice's payoff will be
  \[
    \begin{split}
      0 & = \E\left[ (v - b_A) \I{b_A > \beta_L(v)} \right] \\
        & = \E\left[  v - b_A | v < \beta_L^{-1}(b_A) \right] P(v < \beta_L^{-1}(b_A)) \\
        & = \left( \E\left[  v | v < \beta_L^{-1}(b_A) \right] - b_A \right)
          P(v < \beta_L^{-1}(b_A)). \\
    \end{split}
  \]
  Therefore, we must have, for any $b_A > L$,
  \[
    b_A = \E\left[  v | v < \beta_L^{-1}(b_A) \right].
  \]
  \end{proof}

  \subsection{Proof of Lemma \ref{lemma:bob_bids}}
  \begin{proof}
    We do this in two steps. First, consider the case in which $L=0$. In this case, Bob's bid is:
  \begin{align}\label{eq:bob_bid_0}
      \beta_0(v) = \mathbb{E}[\tilde{v}| \tilde{v} \leq v].
  \end{align}
  Note that we can write $\beta_0(v)$ as:
  \begin{align}
      \beta_0(v) &= \int_0^v x\frac{dF_v(x)}{F_v(v)},\notag\\
       & = v - \frac{\int_0^v F_v(x) dx}{F_v(v)}, \label{eq:beta_v}
  \end{align}
  where the latter equality follows from integration by parts. Note that then we have that
  \begin{align}\label{eq:beta_prime}
  \beta_0'(v) = \frac{f_v(v) \int_0^v F_v(x) dx}{(F_v(v))^2}.
  \end{align}
  To ensure that Bob bids in the way described by $\beta_0(\cdot)$, Alice must bid according to a distribution $F_a, f_a$ such that:
  \begin{align*}
      &\beta_0(v) \in \arg \max_{b \in [0,1]} (v- b)F_a(b),\\
      \implies &(v- \beta_0(v)) f_a(\beta_0(v)) - F_a(\beta_0(v))  = 0, \\
      \implies & v = \beta_0(v) + \frac{F_a(\beta_0(v))}{f_a(\beta_0(v))}.
  \end{align*}

  Let's define a distribution $H(v) = F_a(\beta_0(v))$. By definition $H(v)$ is a valid CDF on the reals. Note that in this case $h(v) = H'(v) = f_a(\beta_0(v)) \beta_0'(v)$. Therefore we implicitly need to find $H(v)$ such that:
  \begin{align*}
          &v - \beta_0(v) = \frac{H(v)\beta_0'(v)}{h(v)}. \\
          \implies & \frac{dH(v)}{H(v)} = \frac{\beta_0'(v) dv}{v- \beta_0(v)}.\\
          \implies & H(v) = \exp \left( \int_0^v \frac{\beta_0'(x) dx}{x-\beta_0(x) } + C\right).
  \intertext{Note that since $\lim _{v \to \infty} H(v) = 1$, $C = - \int_0^\infty \frac{\beta_0'(x) dx}{x-\beta_0(x) }$.}
  \implies & H(v) = \exp \left(  - \int_v^\infty \frac{\beta_0'(x) dx}{x-\beta_0(x) } \right).
  \end{align*}
  Now substituting in \eqref{eq:beta_prime} and \eqref{eq:beta_v} and simplifying, we get that:
  \begin{align*}
      H(v) &= \exp \left(  - \int_v^\infty \frac{\beta_0'(x) dx}{x-\beta_0(x) } \right)\\
   &= \exp \left(  - \int_v^\infty \frac{\frac{f_v(x) \int_0^x F_v(y) dy}{(F_v(x))^2}}{\frac{\int_0^x F_v(y) dy}{F_v(x)}} dx \right)\\
   &= \exp \left(  - \int_v^\infty \frac{f_v(x)}{(F_v(x))} dx \right)\\
   &= F_v(v),
  \end{align*}
  which is exactly the distribution of $H(v)$ that would result under Alice's conjectured strategy.

  Next, consider the case in which $L>0$. In this case, we can think of Alice's strategy as the same as under no limit price, since she bids exactly $L$ in the case of no-limit price with probability zero. The fact that Bob's bidding strategy above is a best-response in the limit-price case then follows straightforwardly: If the best response under no-limit price would have been strictly larger than $L$, then the limit-price is irrelevant. If instead the best response would have been lower than $L$ but the price is larger than $L$, then it is trivially a best response to bid $L$. This concludes the proof of the theorem.
  \end{proof}

  \subsection{Proof of Theorem \ref{theorem:bob_profit_0}}
  \begin{proof}
    Note, we can write Bob's expected profit as,
    \[
      \begin{split}
        \Pi_B = \int_0^\infty \left(v - \beta_0(v)\right) H(v) f_v(v) \, dv.
      \end{split}
    \]
    Where (from the proof of Theorem \ref{theorem:equilibrium}) $H(v) = F_a(\beta_0(v))$ is the probability that
    Bob wins, given that he observes value $v$.  Substituting in $H(v) = F_v(v)$ (from the proof of Theorem
    \ref{theorem:equilibrium}) and $\beta_0(v)$ from \eqref{eq:beta_v}, we get that
    \[
      \begin{split}
        \Pi_B& = \int_0^\infty \left(v - \beta_0(v)\right) F_v(v) f_v(v) \, dv \\
             &= \int_0^\infty \left(\int_0^v F_v(x) \, dx \right) f_v(v) \, dv.
      \end{split}
    \]
    Doing integration by parts,
    \[
      \begin{split}
        \Pi_B &=  \left. \left(\int_0^v F_v(x) \, dx\right) F_v(v) \right\vert_{v=0}^\infty -
                \int_0^\infty F_v(x)^2 \, dx\\
              &= \int_0^\infty F_v(x) \, dx  - \int_0^\infty F_v(x)^2 \, dx\\
              &= \int_0^\infty \left(1-F_v(x)^2\right) \, dx
                - \int_0^\infty (1-F_v(x)) \, dx.
      \end{split}
    \]
    Using the fact that $F_v(\cdot)^2$ is the CDF of $\max(v_1,v_2)$, where $v_1,v_2\sim F_v$ are
    i.i.d.\ draws,
    \[
      \begin{split}
        \Pi_B &=  \E\left[ \max(v_1,v_2) \right] - \E[v_2] \\
              & = \E\left[ \max(v_1 - v_2,0) \right],
      \end{split}
    \]
    as desired.
    \end{proof}

\subsection{Proof of Theorem \ref{theorem:bob_profit_L}}
\begin{proof}
    Bob's expected profit is given by
    \[
      \begin{split}
        \Pi_B & = \int_L^\infty \left(v- \beta_L(v)\right) F_a(\beta_L(v)) f_v(v) \, dv \\
              & = \int_{\underline{v}}^\infty \left(v- \beta_0(v)\right) H(v) f_v(v) \, dv
                +  \int_L^{\underline{v}} (v- L) H(\underline{v}) f_v(v) \, dv, \\
      \end{split}
    \]
    where, for the first equality, we use the observation that $F_a(\beta_L(v))$ is the probability
    that Bob wins given he observes $v$, and for the second equality we use \eqref{eq:bob_bid} and
    $H(v) = F_a(\beta_0(v))$ from the proof of Theorem \ref{theorem:equilibrium}. Then,
    substituting in $H(v) = F_v(v)$ (from the proof of
    Theorem \ref{theorem:equilibrium}) and $\beta_0(v)$ from \eqref{eq:beta_v},
    \[
      \begin{split}
        \Pi_B
    & = \int_{\underline{v}}^\infty \left(v- \beta_0(v)\right) F_v(v) f_v(v) \, dv \\
              & \quad
                +  \left(\E[v| L < v < \underline{v}]- L\right)
                F_v(\underline{v}) (F_v(\underline{v})- F_v(L)) \\
              &= \int_{\underline{v}}^\infty \left(\int_0^v F_v(x) \, dx\right) f_v(v) \, dv +
                \E[v - L | L < v < \underline{v}] F_v(\underline{v})
                (F_v(\underline{v})- F_v(L)).\\
      \end{split}
    \]
    The first term can be simplified via integration by parts as
    \[
      \begin{split}
        \int_{\underline{v}}^\infty \left(\int_0^v F_v(x) \, dx\right) f_v(v) \, dv
        &= \left. \left(\int_0^v F_v(x) \, dx\right) F_v(v) \right\vert_{v=\underline{v}}^\infty -
          \int_{\underline{v}}^\infty F_v(x)^2 \, dx \\
        &= \int_0^\infty F_v(x) \, dx - F_v(\underline{v}) \int_0^{\underline{v}} F_v(x) \, dx -
          \int_{\underline{v}}^\infty F_v(x)^2 \, dx \\
    &= \int_{\underline{v}}^{\infty} F_v(x) (1-F_v(x)) \, dx + (1-F_v(\underline{v}))
       \int_0^{\underline{v}} F_v(x) \, dx.
      \end{split}
    \]
    Combining, we have that
    \[
      \begin{split}
        \Pi_B & =
      \int_{\underline{v}}^{\infty} F_v(x) (1-F_v(x))\, dx + (1-F_v(\underline{v}))  \int_0^{\underline{v}} F_v(x)\, dx \\
      & \quad +\E[v - L| L < v < \underline{v}] F_v(\underline{v}) (F_v(\underline{v})- F_v(L)).
      \end{split}
    \]
    Observe that as $L \to 0$, the latter terms vanish and we are left with the first term evaluated
    at $\underline{v}=0$, i.e., the payoff of an exchange option as in Theorem \ref{theorem:bob_profit_0}.

    Let us go through this term by term. Integration by parts on the first term yields
    \[
      \begin{split}
        \int_{\underline{v}}^{\infty} F_v(x) (1-F_v(x))\, dx
        & =
          \left. x F_v(x) (1-F_v(x)) \right\vert_{x=\underline{v}}^\infty
          - \int_{\underline{v}}^\infty x (1- 2F_v(x)) f_v(x)\, dx\\
        &= -\underline{v} F_v(\underline{v}) (1-F_v(\underline{v}))
          - \int_{\underline{v}}^\infty x (1- 2F_v(x)) f_v(x)\, dx \\
        &= -\underline{v} F_v(\underline{v}) (1-F_v(\underline{v}))
          + \E\left[ \max(v_1,v_2) \I{\max(v_1,v_2) \geq \underline{v}} \right]
        \\
        & \quad
          - \E\left[ v \I{v \geq \underline{v}} \right].
      \end{split}
    \]
    Similarly, the second term can be rewritten as
    \begin{equation}\label{eq:bob-2}
      (1-F_v(\underline{v}))  \int_0^{\underline{v}} F_v(x)\, dx
    =  (1-F_v(\underline{v})) F_v(\underline{v})\E[\underline{v} - v| v < \underline{v}].
    \end{equation}
    Combining, and cancelling terms, we have that
    \[
      \begin{split}
        \Pi_B
        & = \E\left[\max\left(v_1 \I{v_1 \geq \underline{v}} - v_2 \I{v_2 \geq \underline{v}},
          0\right)\right]
        \\
        & \quad
          - (1-F_v(\underline{v})) \E[v \I{v< \underline{v}}]
        \\
        &\quad  + F_v(\underline{v}) \E[(v - L)\I{L < v < \underline{v}}]
      \end{split}
    \]
    Using the definition of $\underline{v}$,
    \[
      \E[v \I{v< \underline{v}}] = L F_v(\underline{v}).
    \]
    Then,
    \[
      \begin{split}
        \Pi_B
        & = \E\left[\max\left(v_1 \I{v_1 \geq \underline{v}} - v_2 \I{v_2 \geq \underline{v}},
          0\right)\right]
        \\
        & \quad
          + F_v(\underline{v}) \left(  \E[(v - L)\I{L < v < \underline{v}}] - (1-F_v(\underline{v})) L \right)
      \end{split}
    \]
    \end{proof}

    \subsection{Proof of Theorem \ref{theorem:uncertain-timing}}

    First note that as in Lemma \ref{lemma:no_pure_strategies}, there is no equilibrium in pure strategies except for the degenerate case where Alice bids $0$ with probability $1$ (possible, as we shall see below, for some settings where $\alpha$ is sufficiently high). To see why, note that if Alice bids a pure strategy strictly above $0$ with probability $1$, then Bob's best response is to exactly match her bid (recall the tiebreaking rule) when profitable to do so, and pass otherwise. Note that in this case Alice only profits when Bob does not act in time, and makes a loss otherwise, so bidding $0$ would have been a strictly better strategy.

    Next note that in any mixed strategy equilibrium for Alice, the lowest bid in the support of her mixed strategy must be exactly $0$. To see why note that otherwise, this lowest bid would again be unprofitable by the same argument as above.

    Note that since she loses for sure whenever Bob acts in time, Alice's expected payoff when she bids $0$ is $\alpha (\mathbb{E}[v])$. Therefore this must be her payoff over all bids in the support of her mixed strategy in equilibrium.

    As in Lemma \ref{lemma:alice_bids}, now suppose that Bob's bid as a function of the observed value $v$ is given by $\beta(v)$. Then, we must have that for any bid $b_A$ in the support of Alice's mixed strategy, we must have that:
    \begin{align*}
            &\alpha (\mathbb{E}[v])  = \E\left[ (v - b_A) \I{b_A > \beta(v)} \right] \\
            & \hphantom{(\alpha (\mathbb{E}[v]))}= \E\left[  v - b_A | v < \beta^{-1}(b_A) \right] P(v < \beta^{-1}(b_A)) \\
            & \hphantom{(\alpha (\mathbb{E}[v]))}= \left( \E\left[  v | v < \beta^{-1}(b_A) \right] - b_A \right)
              P(v < \beta^{-1}(b_A)). \\
    \implies   & b_A  = \E\left[  v | v < \beta^{-1}(b_A) \right] - \frac{\alpha (\mathbb{E}[v])}{P(v < \beta^{-1}(b_A))}.
    \end{align*}
    Doing again a change of variables so that $v$ is the value at which Bob would bid $b_a$, we have that:
    \begin{align*}
        \beta(v) = \E\left[ \tilde{v} | \tilde{v} < v \right] - \frac{\alpha (\mathbb{E}[\tilde{v}])}{P(\tilde{v} < v)}.
    \end{align*}
    Therefore there must be a largest $\bar{v}_\alpha >0$ such that $\beta(v)=0$ for all $v\leq \bar{v}_\alpha$, where $\bar{v}_\alpha$ satisfies:
    \begin{align*}
        \alpha (\mathbb{E}[v]) = \E\left[  v | v < \bar{v}_\alpha \right] \mathbb{P}(v < \bar{v}_\alpha).
    \end{align*}

    Note that given the conjectured strategies for Alice and Bob, it is clear that Alice is indifferent over all bids in the support of her mixed strategy. We are, therefore, only left to check that Bob's strategy is optimal given Alice's.

    Now, note that we can write $\beta(v)$ as:
      \begin{align}
          \beta(v) &= \int_0^v x\frac{dF_v(x)}{F_v(v)} - \frac{\alpha (\mathbb{E}[v])}{F_v(v)},\notag\\
           & = v - \frac{\int_0^v F_v(x) dx}{F_v(v)}- \frac{\alpha (\mathbb{E}[v])}{F_v(v)}. \label{eq:beta_v2}
      \end{align}
      Where the latter equality follows from integration by parts. Note that then we have that
      \begin{align}\label{eq:beta_prime2}
      \beta'(v) = \frac{f_v(v) \int_0^v F_v(x) dx - \alpha (\mathbb{E}[v])}{(F_v(v))^2}.
      \end{align}
      To ensure that Bob bids in the way described by $\beta(\cdot)$, Alice must bid according to a distribution $F_a, f_a$ such that:
      \begin{align*}
          &\beta(v) \in \arg \max_{b \in [0,1]} (v- b)F_a(b),\\
          \implies &(v- \beta(v)) f_a(\beta(v)) - F_a(\beta(v))  = 0, \\
          \implies & v = \beta(v) + \frac{F_a(\beta(v))}{f_a(\beta(v))}.
      \end{align*}

      Let's define a distribution $H(v) = F_a(\beta(v))$. By definition $H(v)$ is a valid CDF on the reals. Note that in this case $h(v) = H'(v) = f_a(\beta(v)) \beta'(v)$. Therefore we implicitly need to find $H(v)$ such that:
      \begin{align*}
              &v - \beta(v) = \frac{H(v)\beta'(v)}{h(v)}. \\
              \implies & \frac{dH(v)}{H(v)} = \frac{\beta'(v) dv}{v- \beta(v)}.\\
              \implies & H(v) = \exp \left( \int_0^v \frac{\beta'(x) dx}{x-\beta(x) } + C\right).
      \intertext{Note that since $\lim _{v \to \infty} H(v) = 1$, $C = - \int_0^\infty \frac{\beta'(x) dx}{x-\beta(x) }$.}
      \implies & H(v) = \exp \left(  - \int_v^\infty \frac{\beta'(x) dx}{x-\beta(x) } \right).
      \end{align*}
      Now substituting in \eqref{eq:beta_prime2} and \eqref{eq:beta_v2} and simplifying, we get that:
      \begin{align*}
          H(v) &= \exp \left(  - \int_v^\infty \frac{\beta'(x) dx}{x-\beta(x) } \right)\\
       &= \exp \left(  - \int_v^\infty \frac{\frac{f_v(x) \int_0^x F_v(y) dy-  f_v(x)\alpha (\mathbb{E}[v])}{(F_v(x))^2}}{\frac{\int_0^x F_v(y) dy}{F_v(x)}- \frac{f_v(x)\alpha (\mathbb{E}[v])}{F_v(x)}} dx \right)\\
       &= \exp \left(  - \int_v^\infty \frac{f_v(x)}{(F_v(x))} dx \right)\\
       &= F_v(v),
      \end{align*}
      which is exactly the distribution of $H(v)$ that would result under Alice's conjectured strategy.


}

\end{document}